\documentclass[12pt,draftcls,onecolumn]{IEEEtran}

\usepackage{amssymb}
\usepackage{amsmath}
\usepackage{color}
\usepackage{cite}
\usepackage{graphicx}
\usepackage{url}
\usepackage{setspace}

\IEEEoverridecommandlockouts

\newtheorem{thm}{Theorem}
\newtheorem{lem}{Lemma}
\newtheorem{mydef}{Definition}
\newtheorem{cor}{Corollary}
\newtheorem{ex}{Example}
\newtheorem{asm}{Assumption}
\newtheorem{prop}{Proposition}
\newtheorem{rem}{Remark}

\title{Competitive Diffusion in Social Networks: \\ Quality or Seeding?}

\author{Arastoo Fazeli$^\dagger$  \quad Amir Ajorlou $^\dagger$  \quad Ali Jadbabaie$^\dagger$
\thanks{\noindent$^\dagger$Department of Electrical and Systems Engineering and GRASP Laboratory at University of Pennsylvania. {\tt\small arastoo@seas.upenn.edu}, {\tt\small ajorlou@seas.upenn.edu} and {\tt\small jadbabai@seas.upenn.edu.} This research was supported by ARO MURI W911NF-12-1-0509 and AFOSR Complex Networks Program.}
\thanks{ A preliminary version of this paper has appeared in \cite{fazeli2014}.}
}

\begin{document}

\maketitle
\thispagestyle{empty}
\pagestyle{empty}

%%%%%%%%%%%%%%%%%%%%%%%%%%%%%%%%%%%%%%%%%%%%%%%%%%%%%%%%%%%%%%%%%%%%%%%%%%%%%%%%

\begin{abstract}

In this paper, we study a strategic model of marketing and product consumption in social networks. We consider two firms in a market competing to maximize the consumption of their products. Firms have a limited budget which can be either invested on the quality of the product or spent on initial seeding in the network in order to better facilitate spread of the product. After the decision of firms, agents choose their consumptions following a myopic best response dynamics which results in a local, linear update for their consumption decision. We characterize the unique Nash equilibrium of the game between firms and study the effect of the budgets as well as the network structure on the optimal allocation. We show that at the equilibrium, firms invest more budget on quality when their budgets are close to each other. However, as the gap between budgets widens, competition in qualities becomes less effective and firms spend more of their budget on seeding. We also show that given equal budget of firms, if seeding budget is nonzero for a balanced graph, it will also be nonzero for any other graph, and if seeding budget is zero for a star graph it will be zero for any other graph as well.
As a practical extension, we then consider a case where products have some preset qualities that can be only improved marginally. At some point in time, firms learn about the network structure and decide to utilize a limited budget to mount their market share by either improving the quality or new seeding some agents to incline consumers towards their products. We show that the optimal budget allocation in this case simplifies to a threshold strategy. Interestingly, we derive similar results to that of the original problem, in which preset qualities simulate the role that budgets had in the original setup.

\end{abstract}

%%%%%%%%%%%%%%%%%%%%%%%%%%%%%%%%%%%%%%%%%%%%%%%%%%%%%%%%%%%%%%%%%%%%%%%%%%%%%%%%

\section{Introduction}

Many recent studies have documented the role of social networks in individual purchasing decisions
\cite{feick1987market,reingen1984brand,godes2004using}. More data from online social networks and advances in information technologies have drawn the attention of firms to exploit this information for their marketing goals. As a result, firms have become more interested in models of influence spread in social networks in order to improve their marketing strategies. In particular, considering the relationship between people in social networks and their rational choices, many retailers are interested to know how to use the information about the dynamics of the spread in order to maximize their product consumption and achieve the most profit in a competitive market.

A main feature of product consumption in these settings is what is often called the ``network effect'' or positive externality. For such products, consumption of each agent incentivizes the neighboring agents to consume more as well, as the consumption decisions between agents and their neighbors are strategic complements of each other.
There are diverse sets of examples for such products or services. New technologies and innovations, mobile applications (e.g., Viber, WhatsApp), online games (e.g., Warcraft), social network web sites (e.g., Facebook, Twitter) and online dating services (e.g., Zoosk, Match.com, OkCupid) are among many examples in which consuming from a common product or service is more preferable for people.

Also, a main property of many products is substitution. A substitute product is a product or service that satisfies the need of a consumer that another product or service fulfills (e.g. Viber and WhatsApp or Gmail and Yahoo email accounts). In all these examples, firms might be interested to utilize the social network among consumers and the positive externality of their products and services to incentivize a larger consumption of their products compared to rival substitute products. Therefore, it is important for firms to know how to shape their strategies in designing their products and offering them to a set of people in order to promote their products intelligently, and eventually achieving a larger share in the market.

In this paper, we study strategic competition between two firms trying to maximize their product consumption. Firms simultaneously allocate their fixed budgets between seeding a set of costumers embedded in a social network and designing the quality of their products. The consumption of each agent is the result of its myopic best response to the previous actions of its peers in the network. Therefore, a firm should provide enough incentives for spread of its product through the payoff that agents receive by consuming it. For this purpose and considering their budgets, firms should strategically design their products and know how to initially seed the network.

We model the above problem as a fixed-sum game between firms, where each firm tries to maximize discounted sum of its product consumption over time, considering its fixed budget. We describe the unique Nash equilibrium of the game between firms which depends on the budgets of the firms and the network structure. We show that at the Nash equilibrium, firms spend more budget on quality when their budgets are close. However, as the difference between budgets increases, firms spend more budget on seeding. We also show that given equal budget of firms, if seeding budget is nonzero for a balanced graph it will also be nonzero for any other graph, and if seeding budget is zero for a star graph it will be zero for any other graph too. Next, we study a different scenario in which firms produce products with some preset qualities. At some point in time, firms learn about the network structure and dedicate some budget to increase their product consumption. The budget can be spent on new seeding of agents in the social network and marginally improving the quality of the products. We derive a simple rule for optimal allocation of the budget between improving the quality and new seeding which in particular depends on the network structure and preset qualities of the products.
%shifted by an offset
%%to compensate for
%proportional to
%the price difference. We refer to these modified qualities as effective qualities hereafter.
%and parameters of the model.
We show that the optimal allocation of the budget depends on the entire centrality distribution of the graph. Specially, we show that maximum seeding occurs in a graph with maximum number of agents with centralities above a certain threshold. Also, the difference in qualities of firms plays an important role in the optimal allocation of the budget. In particular, we show that as the gap between the qualities of the products widens, the firms allocate more budget to seeding. We see that the budgets in the first scenario and preset qualities of the second scenario play similar roles in the optimal allocation.

It is worthwhile to note that the problem of influence and spread in networks has been extensively studied in the past few years \cite{ballester2006s,bharathi2007competitive,galeotti2009influencing,kempe2003maximizing,kempe2005influential,chasparis2010control,vetta2002nash}. Also, diffusion of new behaviors and strategies through local coordination games has been an active field of research
\cite{ellison1993learning,kandori1993learning,harsanyi1988general,young1993evolution,young2001individual,young2002diffusion,montanari2010spread,kleinberg2007cascading}. Goyal and Kearns proposed a game theoretic model of product adoption in \cite{goyal2012competitive}. They computed upper bounds of the price of anarchy and showed how network structure may amplify the initial budget differences. Similarly, in \cite{bimpikiscompeting} Bimpikis, Ozdaglar and Yildiz proposed a game theoretic model of competition between firms which can target their marketing budgets to individuals embedded in a social network. They provided conditions under which it is optimal for the firms to asymmetrically target a subset of the individuals. Also, Chasparis and Shamma assumed a dynamical model of preferences in \cite{chasparis2010control} and computed optimal policies for finite and infinite horizon where endogenous network influences, competition between two firms and uncertainties in the network model were studied.
The main contribution of our work is to explicitly study the tradeoff between investing on quality of a product and initial seeding in a social network. Our model is similar to the model proposed in \cite{fazeli2012duopoly}, however, instead of pricing strategy in \cite{fazeli2012duopoly}, the notion of quality is introduced and the tradeoff between quality and seeding is studied. Also, our model is tractable and allows us to characterize the exact product consumption at each time, instead of lower and upper bounds provided in \cite{fazeli2012game,fazeli2012targeted}.

%The rest of this paper is organized as follows: In Section~\ref{sec2}, we introduce our model and update dynamics for agents applying the myopic best response. In Section~\ref{sec3}, we study the game played among the firms in which firms decide on quality versus initial seeding. Optimal budget allocation strategy for quality improvement versus new seeding is presented in Section~\ref{sec4}. Finally, in Section~\ref{sec5}, we conclude the paper.

%%%%%%%%%%%%%%%%%%%%%%%%%%%%%%%%%%%%%%%%%%%%%%%%%%%%%%%%%%%%%%%%%%%%%%%%%%%%%%%%

\section{The Spread Dynamics} \label{sec2}

There are $n$ agents $V=\{1,\ldots,n\}$ in a social network. The relationship among agents is represented by a directed graph $\mathcal{G}=(V,E)$ in which agents $i, j \in V$ are neighbors if $(i,j) \in E$. The weighted adjacency matrix of the graph $\mathcal{G}$ is denoted by a row stochastic matrix $G$ where the $ij$-th entry of $G$, denoted by $g_{ij}$, represents the strength of the influence of agent $j$ on $i$. For diagonal elements of matrix $G$, we have $g_{ii}=0$ for all agents $i \in V$. We assume that there are two competing firms $a$ and $b$ producing product $a$ and $b$. Each agent has  one unit demand which can be supplied by either of the firms.
We define the variable $0 \leq x_i(t) \leq 1$ and $0 \leq 1- x_i(t) \leq 1$ as the consumption of the product $a$ and $b$ by agent $i$ at time $t$.

Denote by $q_a, q_b \geq \epsilon > 0 $ the quality of product $a$ and $b$ respectively, where $\epsilon$ has an infinitesimal value. The values of $q_a$ and $q_b$ can be interpreted as the payoff that any two agents $i$ and $j$ would achieve if they both consume the same product. In other words, we can assume $q_a$ and $q_b$ are payoffs of the following game
\begin{center}
\begin{tabular}{|l|c|r|}
  \hline
   & $x_j$ & $1-x_j$ \\ \hline
  $x_i$ & $q_a x_i x_j$ & $0$ \\ \hline
  $1-x_i$ & $0$ & $q_b (1-x_i)(1-x_j)$ \\
  \hline
\end{tabular}
\end{center}
Since agents benefit from the same action of their neighbors, this game could be thought of as a local coordination game. From the above table it follows easily that the payoff of agents $i$ and $j$ from their interaction is
\begin{equation*}
u_{ij}(x_i,x_j)=q_a x_i x_j + q_b (1-x_i)(1-x_j).
\end{equation*}
We also assume that each agent benefits from taking action $x_i$ irrespective of actions taken by its neighbors. We assume the isolation payoff of consuming $x_i$ and $1-x_i$ from product $a$ and $b$ is represented by the following quadratic form functions
\begin{equation*}
u_{ii}^{a} = q_a ( \alpha x_i - \beta x^2_i), \qquad \qquad u_{ii}^{b} = q_b [\alpha (1-x_i) - \beta (1-x_i)^2],
\end{equation*}
where $\alpha$ and $\beta$ are parameters of the isolation payoff. This forms of payoff indicates that a product with higher quality has a higher isolation payoff. The total isolation payoff of agent $i$ can be written as
\begin{equation*}
u_{ii}(x_i) = \{ q_a (\alpha x_i - \beta x^2_i)\} + \{q_b \left(  \alpha (1-x_i) - \beta (1-x_i)^2  \right)  \}.
\end{equation*}
%The above isolation payoff indicates that the isolation utility of agent $i$ disappears when $x_i=0$ or $x_i=1$, i.e. the agent consumes only one of the product $a$ or $b$.
In order to have nonnegative isolation payoff for $x_i=0$ and $x_i=1$, we assume $\beta \leq \alpha$. Assuming quadratic form function for the isolation payoff not only makes the analysis more tractable, but also is a good second order approximation for the general class of concave payoff functions. By changing the variables $x_i = \frac{1}{2} + y_i$ after simplification we get
\begin{equation*}
\begin{split}
&u_{ij}(y_i,y_j)=q_a (\frac{1}{2} + y_i) (\frac{1}{2} + y_j) + q_b (\frac{1}{2} - y_i)(\frac{1}{2} - y_j), \\
&u_{ii}(y_i) = (q_a + q_b)(\frac{\alpha}{2} - \frac{\beta}{4} - \beta y^2_i) + (q_a - q_b) (\alpha - \beta) y_i.
\end{split}
\end{equation*}
Therefore, the total utility of agent $i$ from taking action $y_i$ is given by
\begin{equation} \label{total utility}
\begin{split}
&U_i(y_i, \vec{y}_{-i}) = (q_a + q_b)(\frac{\alpha}{2} - \frac{\beta}{4} - \beta y^2_i) + (q_a - q_b) (\alpha - \beta) y_i \\
& + q_a \sum_{j=1}^n g_{ij} (\frac{1}{2} + y_i) (\frac{1}{2} + y_j) + q_b \sum_{j=1}^n g_{ij} (\frac{1}{2} - y_i) (\frac{1}{2} - y_j).
%\\
%& - p_a(\frac{1}{2} + y_i) - p_b(\frac{1}{2} - y_i),
\end{split}
\end{equation}
%\begin{equation} \label{total utility}
%\begin{split}
%&U_i(y_i, \vec{y}_{-i}) = u_{ii}(y_i) + \sum_{j=1}^n g_{ij} u_{ij}(y_i,y_j) - p_a(\frac{1}{2} + y_i) \\
%& - p_b(\frac{1}{2} - y_i),
%%=(q_a + q_b)(\frac{\alpha}{2} - \frac{\beta}{4} - \beta y^2_i) + (q_a - q_b) (\alpha - \beta) y_i \\
%%& + q_a \sum_{j=1}^n g_{ij} (\frac{1}{2} + y_i) (\frac{1}{2} + y_j) + q_b \sum_{j=1}^n g_{ij} (\frac{1}{2} - y_i) (\frac{1}{2} - y_j) \\
%%& - p_a(\frac{1}{2} + y_i) - p_b(\frac{1}{2} - y_i),
%\end{split}
%\end{equation}
%where $0 \leq \alpha \leq 1$ represents the weight that each agent puts on its isolation payoff compared to the payoff from its neighbors.
%where $p_a$ and $p_b$ are the prices of each unit of product $a$ and $b$ respectively.
In the above equation $\vec{y}_{-i}$ denotes an action vector of all agents other than agent $i$. From equation \eqref{total utility} we can see that product $a$ and $b$ have a positive externality effect in the network, meaning that the usage level of an agent has a positive impact on the usage level of its neighbors. Therefore, it follows that $q_a$ and $q_b$ in addition to the payoff of a local coordination game, can be interpreted as coefficients of network externality of product $a$ and $b$ respectively.

We assume agents repeatedly apply myopic best response to the actions of their neighbors. This means that each agent, considering its neighbors consumptions at the current period, chooses the amount of the product that maximizes its current payoff, as its consumption for the next period. In other words, consumption of agent $i$ at time $t+1$ is updated as follows
\begin{equation*} \label{utility}
y_i(t+1) = \arg\max_{y_i}\quad U_i(y_i(t), \vec{y}_{-i}(t)).
\end{equation*}
The above equation results in the following update dynamics
\begin{equation*} \label{best response}
%\begin{split}
y_i(t+1)= (\frac{1}{2\beta})\sum_{j=1}^n g_{ij} y_j(t) + (\frac{q_a-q_b}{4\beta(q_a+q_b)})\sum_{j=1}^n g_{ij}
 + (\frac{(\alpha - \beta)(q_a - q_b)}{2\beta(q_a+q_b)}).
% - (\frac{p_a-p_b}{2\beta(q_a+q_b)}) .
%\end{split}
\end{equation*}
Therefore, the consumption of the product $a$ can be written as the following linear update dynamics form
%\begin{small}
\begin{equation} \label{update a}
%\begin{split}
\vec{y}(t+1)= (\frac{1}{2\beta}) G \vec{y}(t) + (\frac{(1+2(\alpha - \beta))(q_a - q_b) }{4\beta(q_a+q_b)})\vec{1}.
%\end{split}
\end{equation}
%\end{small}
Similarly, for the consumption of the product $b$ we have
$1 - x_i(t) = \frac{1}{2} - y_i(t)$.
%If we define effective quality of products $a$ and $b$ by
%\begin{equation} \label{effective quality}
%%\begin{split}
%q_a = q_a - \frac{p_a-p_b}{(1+2(\alpha - \beta))}, \qquad \qquad
%q_b = q_b - \frac{p_b-p_a}{(1+2(\alpha - \beta))},
%%\end{split}
%\end{equation}
%then the update dynamics in \eqref{update a} can be written as
%\begin{equation} \label{update a}
%\vec{y}(t+1)= (\frac{1}{2\beta}) G \vec{y}(t) + (\frac{(1+2(\alpha - \beta))(q_a - q_b)}{4\beta(q_a+q_b)})\vec{1}.
%\end{equation}
%%\begin{equation} \label{update a}
%%\vec{y}(t+1)= (\frac{1-\alpha}{2\alpha}) G \vec{y}(t) + (\frac{q_a-q_b}{4\beta(q_a+q_b)}) \vec{1}.
%%\end{equation}
%$q_a$ and $q_b$ are the qualities shifted by an offset to compensate for the price differences. We can interpret $\tilde{q}$ as the attractiveness or the effective quality of the products for agents in the network. As it can be seen from equation \eqref{effective quality}, when the product $a$ is priced higher than product $b$ then the effective quality of product $a$ for agents goes down.

%Also, when $(\alpha - \beta)$ is higher, which results in a higher isolation payoff for agents, then the effective quality is higher too.
\begin{asm} \label{asm1}
We assume $1 + \alpha \leq 2\beta$.
%, i.e. the self isolated payoff has higher weight than the payoff from neighbors.
This assumption guaranties that $0 \leq x_i(t) \leq 1$ for all $i$ and all $t$ under the update rule \eqref{update a}.
%We also assume that effective qualities are positive, i.e. $q_a, q_b >0$.
\end{asm}
Using the above assumption \label{asm1} and defining
\begin{equation} \label{definition}
W \triangleq (\frac{1}{2\beta}) G, \qquad \qquad \vec{u}_a \triangleq \left(  ( \frac{1+2(\alpha - \beta)}{4\beta} ) (\frac{q_a - q_b}{q_a + q_b} )   \right)\vec{1},
\end{equation}
equations \eqref{update a} can be written as
\begin{equation*} \label{matrix form}
\vec{y}(t+1)= W \vec{y}(t) + \vec{u}_a.
\end{equation*}
%Equation \eqref{matrix form}
The above equation can be expanded as
\begin{equation} \label{expanded form}
\vec{y}(t)= W^t \vec{y}(0) + \sum_{k=0}^{t-1} W^k \vec{u}_a.
\end{equation}
Therefore, the consumption of agents depends on the initial preferences, i.e. $\vec{y}(0)$, the quality of product $a$ and $b$, i.e. $q_a$ and $q_b$, and the structure of the network, i.e. the matrix $G$. In the next section we discuss how firms can exploit this information in order to maximize their product consumption and also characterize the unique Nash equilibrium of the game played between two firms.

%\section{Nash Equilibrium Analysis}
\section{Optimal Budget Allocation} \label{sec3}

%In this section we study how firms can maximize the spread of their products and also describe the Nash equilibrium of the game played between these two firms. As we mentioned earlier, the goal of these two competing firms is to maximize the adoption of their products in the network. However, they will have to consider their budget as well. We assume that each firm divides its budget between quality and initial seed of its product. We assume that the utility of each firm is the discounted sum of its product adoption over time. Hence, firms solve following optimization problems

%In this section we describe the game between firms where each firm maximizes the consumption of its product given a fixed budget. Initially firms receive a fixed budget that they can either invest on ``designing the quality'' of their products or spend it on initial seedings of some agents or both.
In this section we describe the game between firms where each firm aims to maximize the consumption of its product over an infinite time horizon given a fixed budget. Each firm has an initial budget that it can either invest on ``quality'' or spend it on promoting its product by seeding some of the agents, or both.
This initial seeding can be viewed as free offers to promote the products in the network.
We define the utility of each firm as the discounted sum of its product consumption over time
\begin{equation*}
\begin{split}
&U_a = \sum_{t=0}^{\infty} \delta^t \vec{1}^T ((0.5)\vec{1} + \vec{y}(t)), \\
&U_b = \sum_{t=0}^{\infty} \delta^t \vec{1}^T ((0.5)\vec{1} - \vec{y}(t)).
\end{split}
\end{equation*}
%As it can be seen from equations above, firms play a fixed-sum game.
Each firm has a limited budget $K_a, K_b$ that can spend on either initial seeding, i.e. $\vec{S}_a$ and $\vec{S}_b$, or designing the quality of its product, i.e. $q_a$ and $q_b$, or both. Seeding $\vec{S}_a$ and $\vec{S}_b$ will change the initial consumption of products $a$ and $b$ by $\vec{S}_a-\vec{S}_b$ and $\vec{S}_b-\vec{S}_a$ respectively.
Therefore, the amount that agents initially consume from product $a$ and $b$ will be $\vec{x}(0) = (0.5)\vec{1}  + \vec{S}_a - \vec{S}_b$ and $\vec{1} - \vec{x}(0) = (0.5)\vec{1} + \vec{S}_b - \vec{S}_a$.
This means that if both firms seed an agent equally then the agent has no preference for one product over the other, i.e. $\vec{y}(0)=\vec{0}$. This assumption can be justified since agents should be initially indifferent between products before their consumption and realizing the quality of products if initial seedings by firms are equal. In order to have $0 \leq x_i(0) \leq 1$ and $0 \leq 1 - x_i(0) \leq 1$ for all agents $i$, we impose the constraints $\|\vec{S}_a\|_\infty\leq 0.5$ and $\|\vec{S}_b\|_\infty\leq 0.5$. This means that firms can initially seed agents up to their demand capacity which is $0.5$ for all agents. Using equations \eqref{definition} and \eqref{expanded form} and defining the centrality vector $v$ by $v= (I-\delta W^T)^{-1} \vec{1}$ where agents are ordered from the highest to the lowest centrality, i.e. $v_1 \geq v_2 \geq \cdots \geq v_n$,
and noting that $\sum v_i = \frac{2\beta n}{2\beta- \delta}$,
the utilities of firms can be written as
\begin{equation} \label{payoff}
\begin{split}
&U_a = (\frac{n}{2(1-\delta)}) + v^T \vec{S}_a - v^T \vec{S}_b  + \lambda (\frac{q_a-q_b}{q_a+q_b}), \\
&U_b = (\frac{n}{2(1-\delta)}) + v^T \vec{S}_b - v^T \vec{S}_a  + \lambda (\frac{q_b-q_a}{q_a+q_b}),
\end{split}
\end{equation}
where
\begin{equation} \label{lambda}
\lambda = \frac{\delta(1+2(\alpha-\beta)) n}{2(1-\delta)(2\beta- \delta)}.
\end{equation}
We assume the cost of each unit of quality is given by $c_q$ and the cost of each unit of initial seeding is given by $c_s$. Therefore, the game between the firms can be written as
\begin{equation*} \label{payoffs}
\begin{split}
& \max_{\vec{S}_a, q_a} \quad (\frac{n}{2(1-\delta)}) +  v^T \vec{S}_a - v^T \vec{S}_b + \lambda (\frac{q_a-q_b}{q_a+q_b}), \\
& \text{s.t.} \quad c_s \|\vec{S}_a\|_1 + c_q q_a=K_a,
\end{split}
\end{equation*}
for firm $a$, and
\begin{equation*} \label{payoffs}
\begin{split}
& \max_{\vec{S}_b, q_b} \quad (\frac{n}{2(1-\delta)}) + v^T \vec{S}_b - v^T \vec{S}_a + \lambda (\frac{q_b-q_a}{q_a+q_b}), \\
& \text{s.t.} \quad c_s \|\vec{S}_b\|_1 + c_q q_b=K_b,
\end{split}
\end{equation*}
for firm $b$.
%Here for firm $a$ (or $b$), one of variables $q_a$ or $\|\vec{S}_a\|_1$ is the independent variable and the other one is determined such that $c_s \|\vec{S}_a\|_1 + c_q q_a=K_a$.
Since the effect of the action of $\vec{S}_b$ is decoupled from $\vec{S}_a$ in $U_a$, therefore, the optimization problem of firm $a$ is equivalent to
\begin{equation*}
\begin{split}
& \max_{\vec{S}_a, q_a} \quad v^T \vec{S}_a + \lambda (\frac{q_a-q_b}{q_a+q_b}),\\
& \text{s.t.} \quad c_s \|\vec{S}_a\|_1 + c_q q_a=K_a.
\end{split}
\end{equation*}
Similarly, for firm $b$ we have
%\begin{equation*}
%\begin{split}
%& \max_{\vec{S}_a, q_a} \quad  v^T \vec{S}_a + \lambda (\frac{q_a-q_b}{q_a+q_b}), \qquad \qquad \max_{\vec{S}_b, q_b} \quad v^T \vec{S}_b + \lambda (\frac{q_b-q_a}{q_a+q_b}),\\
%& \text{S.t.} \quad c_s \|\vec{S}_a\|_1 + c_q q_a=K_a \qquad \qquad \text{S.t.} \quad c_s \|\vec{S}_a\|_1 + c_q q_a=K_a.
%\end{split}
%\end{equation*}
\begin{equation*}
\begin{split}
& \max_{\vec{S}_b, q_b} \quad v^T \vec{S}_b + \lambda (\frac{q_b-q_a}{q_a+q_b}),\\
& \text{s.t.} \quad c_s \|\vec{S}_b\|_1 + c_q q_b=K_b.
\end{split}
\end{equation*}
%\subsubsection {Whom to be seeded} \label{seeding part}
%$\hfill$
%As it can be seen the utility function of firm $a$ and $b$ are concave in $q_a$ and $q_b$ and linear concave in $\vec{S}_a$ and $\vec{S}_b$ respectively.
%If we assume $K_a, K_b \leq \frac{n}{2}$, i.e. the total budget of firms is less than seeding capacities of agents,
%It can be easily shown that in an optimal solution the agents with highest centralities are seeded first, i.e. $S_{a_i}^{*}=\frac{1}{2}$ and $S_{b_j}^{*}=\frac{1}{2}$ for all $i<k$ and $j<l$  for some integers $1 \leq k,l \leq n$, until the optimal tradeoff between the total amount of seeding and quality is reached on agent $k$ and $l$, i.e. $S_{a_k}^{*} \leq \frac{1}{2}$ and $S_{b_l}^{*} \leq \frac{1}{2}$, and after that agents are not seeded, i.e. $S_{a_i}^{*}=0$ and $S_{b_j}^{*}=0$ for all $i>k$ and $j>l$. In the following theorem we describe the Nash equilibrium of the game played between firms. To do so, first we mention a lemma which guarantees uniqueness of the Nash equilibrium.
\begin{rem}
\label{actionspace}
It can be easily shown that for a seeding budget $\|S_a\|_1$, the optimal seeding strategy is to seed the agents in the order of their centralities (from highest to lowest) until we either ran out of budget or all the agents are seeded. Therefore, an optimal action $(\vec{S}_a, q_a)$ is fully determined from $(\|S_a\|_1, q_a)$, thus reducing the action space of firm $a$ to only $q_a$, given its budget constraint. Similar argument holds for firm $b$. Therefore, we may look at the utilities $U_a$ and $U_b$ as functions of $(q_a,q_b)$ under the optimal seeding and fixed budgets.
\end{rem}

In order to study the existence and uniqueness of the Nash equilibrium for the above game, we use a variation of the well-known Sion's minimax theorem (see \cite{sion1958general} for the original Sion's theorem) as below.
\begin{lem} \label{lem1}
Consider a two person zero-sum game, on closed, bounded, and convex finite-dimensional action sets $\Omega_1 \times \Omega_2$, defined by the continuous function $L(x_1, x_2)$. Let $L(x_1, x_2)$ be strictly convex in $x_1$ for each $x_2 \in \Omega_2$ and strictly concave in $x_2$ for each $x_1 \in \Omega_1$. Then the game admits a unique pure strategy Nash equilibrium.
\end{lem}
\begin{proof}
See Theorem $A.4$ on page $286$ in \cite{alpcan2010network}.
\end{proof}
In the following theorem we characterize the Nash equilibrium of the game played between firms.
\begin{thm}
Consider firms $a$ and $b$ with utility functions $U_a$ and $U_b$ as described in \eqref{payoff}.
%and $\epsilon < K_a, K_b < \frac{c_s n}{2}$.
%If for payoff of products we assume Cobb-Douglas function, i.e. $q(c)= d l^{\beta} c^{\alpha} $,
%In the unique Nash equilibrium of the game between firms we have
The game between firms admits a unique Nash equilibrium of form
\begin{equation} \label{Nash eq}
\begin{split}
&q_a^*= (2 \lambda)(\frac{c_s}{c_q})(\frac{\tilde{v}_l}{(\tilde{v}_k+\tilde{v}_l)^2}), \qquad \qquad \qquad \qquad \qquad \qquad \quad q_b^*= (2 \lambda)(\frac{c_s}{c_q}) (\frac{\tilde{v}_k}{(\tilde{v}_k+\tilde{v}_l)^2}), \\
&S_{a_i}^{*} = \left\{
               \begin{array}{ll}
                 \frac{1}{2} & \hbox{$1 \leq i < k $}, \\
                 \frac{K_a}{c_s} - \frac{k-1}{2}-\frac{2 \lambda \tilde{v}_l}{(\tilde{v}_k+\tilde{v}_l)^2} & \hbox{$i =  k $}, \\
                 $0$ & \hbox{$i >  k $},
               \end{array}
             \right.  \qquad
S_{b_i}^{*} = \left\{
               \begin{array}{ll}
                 \frac{1}{2} & \hbox{$1 \leq i < l $}, \\
                 \frac{K_b}{c_s} - \frac{l-1}{2}-\frac{2 \lambda \tilde{v}_k}{(\tilde{v}_k+\tilde{v}_l)^2} & \hbox{$i =  l $}, \\
                 $0$ & \hbox{$i >  l $},
               \end{array}
             \right.
\end{split}
\end{equation}
for some $v_{k}  \leq  \tilde{v}_k  \leq  v_{k-1}$ and $v_{l}  \leq  \tilde{v}_l  \leq  v_{l-1}$ that satisfy
\begin{equation} \label{last agent}
\begin{split}
& 0 \leq S_{a_k}^{*} = \frac{K_a}{c_s} - \frac{k-1}{2}-\frac{2 \lambda \tilde{v}_l}{(\tilde{v}_k+\tilde{v}_l)^2} < \frac{1}{2}, \\
& 0 \leq S_{b_l}^{*} = \frac{K_b}{c_s} - \frac{l-1}{2}-\frac{2 \lambda \tilde{v}_k}{(\tilde{v}_k+\tilde{v}_l)^2} < \frac{1}{2},
\end{split}
\end{equation}
where $\tilde{v}_k=v_k$ if $S_{a_k}^{*}>0$ and $\tilde{v}_l=v_l$ if $S_{a_l}^{*}>0$.
\footnote{We define $v_{0} \triangleq \infty$. If $S_{a_n}^{*} = \frac{1}{2}$ or $S_{b_n}^{*} = \frac{1}{2}$, then $\tilde{v}_n  \leq  v_{n}$.}
\end{thm}
\begin{proof}
Given the optimal seeding of each firm, i.e. seeding agents from the highest to the lowest centrality, as discussed in Remark~\ref{actionspace}, the tradeoff between seeding amount and quality can be solved by optimizing $U_a$ and $U_b$ with respect to $q_a$ and $q_b$ respectively. The action space of firms, i.e. $\epsilon\leq q_a\leq\frac{K_a}{c_q}$ and $\epsilon\leq q_b\leq\frac{K_b}{c_q}$, is a closed, bounded, and convex finite-dimensional set. Also, $U_a + U_b = \frac{n}{(1-\delta)}$, hence, the game is a fixed-sum game and can be transformed to a zero sum game by subtracting the constant value of $\frac{n}{2(1-\delta)}$ from $U_a$ and $U_b$.
%which does not affect the optimal solution.
The term $v^T \vec{S}_a$ in $U_a$ is piecewise linear in $\|S_a\|_1$ and thus in $q_a$, under optimal seeding. Using this, it is easy to see that $U_a(q_a, q_b)$ is strictly concave in $q_a$ for each $q_b$, and strictly convex in $q_b$ for each $q_a$ via a similar argument.
%Furthermore, by taking the second derivative it is easy to see that $U_a(q_a, q_b)$ is strictly concave in $q_a$ for each $q_b$ and strictly convex in $q_b$ for each $q_a$.
Therefore, based on Lemma \ref{lem1}, the game admits a unique Nash equilibrium. Assume that the first $(k-1)$ and $(l-1)$ agents are fully seeded by firms $a$ and $b$ respectively at equilibrium. Then, from the budget constraints we have $S_{a_k}^{*} = \frac{K_a}{c_s} - \frac{k-1}{2}-(\frac{c_q}{c_s})q_a$, and $S_{b_l}^{*} = \frac{K_b}{c_s} - \frac{l-1}{2}-(\frac{c_q}{c_s})q_b$, therefore, by plugging in the vector of optimal seeding $S^{*}_a$ and $S^{*}_b$ as described earlier,
%from part \eqref{seeding part},
the optimization problem of firms is given by
\begin{equation*}
\begin{split}
& \max_{\epsilon \leq q_a \leq \frac{K_a}{c_q}} \quad (\frac{1}{2})\sum_{i=1}^{k-1} v_i + (\frac{K_a}{c_s} - \frac{k-1}{2}-(\frac{c_q}{c_s})q_a)v_k + \lambda (\frac{q_a-q_b}{q_a+q_b}), \\
& \max_{\epsilon \leq q_b \leq \frac{K_b}{c_q}} \quad (\frac{1}{2})\sum_{i=1}^{l-1} v_i + (\frac{K_b}{c_s} - \frac{l-1}{2}-(\frac{c_q}{c_s})q_b)v_l + \lambda (\frac{q_b-q_a}{q_a+q_b}).
\end{split}
\end{equation*}
If $0 < S_{a_k}^{*} < \frac{1}{2}$ and $0 < S_{b_l}^{*} < \frac{1}{2}$, the first order optimality condition requires taking the derivative of the two equations above with respect to $q_a$ and $q_b$ and setting them to zero
\begin{equation*} \label{FOC}
\begin{split}
& -(\frac{c_q}{c_s})v_k  +  (\frac{2\lambda q_b}{(q_a+q_b)^2})  = 0, \\
%\qquad \qquad
& -(\frac{c_q}{c_s})v_l  +  (\frac{2\lambda q_a}{(q_a+q_b)^2})  = 0.
\end{split}
\end{equation*}
Solving equations above we get
\begin{equation*}
\begin{split}
& q_a^*= (2\lambda)(\frac{c_s}{c_q})(\frac{v_l}{(v_k+v_l)^2}), \\
%\qquad \quad
& q_b^*= (2\lambda)(\frac{c_s}{c_q}) (\frac{v_k}{(v_k+v_l)^2}),
\end{split}
\end{equation*}
where
%$0 < K_a - \frac{k-1}{2} - q_a^*(k) \leq \frac{1}{2}$ and $0 < K_b- \frac{l-1}{2} - q_b^*(l) \leq \frac{1}{2}$ and
$1 \leq k, l \leq n$ are integers that must satisfy conditions in \eqref{last agent} for $\tilde{v}_k = v_k$ and $\tilde{v}_l = v_l$.
%since $0 < S_{a_k}^{*} \leq \frac{1}{2}$ and $0 < S_{b_l}^{*} \leq \frac{1}{2}$.
If $S_{a_k}^{*} = 0$ and $S_{b_l}^{*} = 0$, the first order optimality condition is as follows
%\begin{equation} \label{FOC2}
%\begin{split}
%&v_{k}  \leq  (2 \lambda)(\frac{c_s}{c_q}) (\frac{q_b}{(q_a+q_b)^2})  \leq  v_{k-1}, \\
%%\qquad \qquad
%&v_{l}  \leq  (2 \lambda)(\frac{c_s}{c_q}) (\frac{q_a}{(q_a+q_b)^2})  \leq  v_{l-1},
%\end{split}
%\end{equation}
\begin{equation} \label{FOC tilde}
v_{k}  \leq  \tilde{v}_k  \leq  v_{k-1}, \qquad \qquad v_{l} \leq  \tilde{v}_l  \leq  v_{l-1},
\end{equation}
where
\begin{equation} \label{tilde}
\tilde{v}_k = (2 \lambda)(\frac{c_s}{c_q}) (\frac{q_b}{(q_a+q_b)^2}), \qquad \qquad \tilde{v}_l = (2 \lambda)(\frac{c_s}{c_q}) (\frac{q_a}{(q_a+q_b)^2}),
\end{equation}
and if $S_{a_n}^{*} = \frac{1}{2}$ or $S_{b_n}^{*} = \frac{1}{2}$ then $\tilde{v}_n  \leq  v_{n}$.
We can solve $q_a^{*}$ and $q_b^{*}$ in terms of $\tilde{v}_k$ and $\tilde{v}_l$ as described in \eqref{Nash eq}.
\end{proof}

\begin{cor}
If firms have equal budgets $K_a=K_b=K$, then in the unique symmetric Nash equilibrium of the game between firms we have
\begin{equation} \label{Symmetric Nash}
q_a^* = q_b^* =(\frac{\lambda}{2})(\frac{c_s}{c_q})(\frac{1}{\tilde{v}_l}), \qquad \qquad
S_{a_i}^{*} = S_{b_i}^{*} = \left\{
               \begin{array}{ll}
                 \frac{1}{2} & \hbox{$1 \leq i < l $}, \\
                 \frac{K}{c_s} - \frac{l-1}{2}-\frac{\lambda} {2\tilde{v}_l} & \hbox{$i =  l $}, \\
                 $0$ & \hbox{$i >  l $},
               \end{array}
             \right.
\end{equation}
for some $v_{l}  \leq  \tilde{v}_l  \leq  v_{l-1}$ that satisfy
%where $1\leq l \leq n$ is a unique integer such that
\begin{equation} \label{last agent symmetric nash}
0 \leq S_{a_l}^{*} = S_{b_l}^{*} = \frac{K}{c_s} - \frac{l-1}{2}-\frac{\lambda} {2\tilde{v}_l} < \frac{1}{2}, \\
\end{equation}
where $\tilde{v}_l=v_l$ if $S_{a_l}^{*} = S_{b_l}^{*} >0$.
\footnote{We define $v_{0} \triangleq \infty$. If $S_{a_n}^{*} = S_{b_n}^{*} = \frac{1}{2}$, then $\tilde{v}_n  \leq  v_{n}$.}
\end{cor}

Equation \eqref{Nash eq} indicates that the Nash equilibrium depends on both the budgets of the firms, i.e. $K_a$ and $K_b$, centrality distribution of agents in the network, i.e. $v$.
%, the parameter $\lambda$ which depends on discounting factor $\delta$ and isolation payoff parameters $\alpha$ and $\beta$, and cost of seeding and quality, i.e. $c_s$ and $c_q$.
We will discuss the effect of each of these factors on the Nash equilibrium in the following subsections. All of our analysis here is for firm $a$ and similar results can be shown for firm $b$ as well. For simplicity,
%when the parameters of the budget constraint $c_s \|\vec{S}_a\|_1 + c_q q_a=K_a$ are fixed, i.e. $c_s$, $c_q$ and $K_a$ are fixed,
we only discuss seeding budget; quality budget can be found easily using the budget constraint.

\subsection{Effect of Budget of Firms on Firms' Decisions:}
In this subsection we study how the budget of each firm, i.e. $K_a$ and $K_b$, can influence the Nash equilibrium.
As it can be seen from \eqref{Nash eq}, the Nash equilibrium depends on both $\tilde{v}_k$ and $\tilde{v}_l$, which in turn depend on firm's and its rival's budgets, i.e. both $K_a$ and $K_b$.
In the first proposition, we compare the seeding budget and quality of two firms at the Nash equilibrium with respect to their budgets. We first prove the following lemma.

\begin{lem} \label{lem3}
At the Nash equilibrium, if $q_a^{*} < q_b^{*}$, then $\|\vec{S}_a^{*}\|_1 \leq \|\vec{S}_b^{*}\|_1$.
\end{lem}
%\begin{proof}
%If $q_a^{*} < q_b^{*}$, then from \eqref{Nash eq}, we have $\tilde{v}_l < \tilde{v}_k$. Therefore, from Lemma \ref{lm1} we have $\|\vec{S}_a^{*}\|_1 \leq \|\vec{S}_b^{*}\|_1$.
%\end{proof}
\begin{proof}
If $q_a^{*} < q_b^{*}$, then from \eqref{Nash eq}, we have $\tilde{v}_l < \tilde{v}_k$. If $ \tilde{v}_{l} < \tilde{v}_k$ then either $k<l$ or $l=k$. If $k<l$ it is obvious to see that $\vec{S}^*_a\leq \vec{S}^*_b$. If $l=k$ then we have two cases: If $0 < S_{a_{k}}^{*} < \frac{1}{2}$, then based on \eqref{FOC tilde} we have $ \tilde{v}_k = v_k = v_{l} \leq \tilde{v}_{l}$ which is a contradiction with $\tilde{v}_{l} < \tilde{v}_k$. If $S_{a_{k}}^{*} = 0$, then obviously $S_{a_k}^{*} \leq S_{b_{l}}^{*}$ and therefore, $\vec{S}^*_a\leq \vec{S}^*_b$. If $S_{a_{n}}^{*} =\frac{1}{2}$, then $\tilde{v}_l < \tilde{v}_k = \tilde{v}_n \leq v_n$, hence, $S_{b_{n}}^{*} =\frac{1}{2}$. This finishes the proof.
\end{proof}
The next proposition states that the firm with higher budget surpasses the rival in both quality and seeding.
\begin{prop} \label{prop compare nash}
The firm with higher budget has higher seeding budget and quality, i.e. if $K_b\leq K_a$, then $\|\vec{S}_b^*\|_1\leq\|\vec{S}_a^*\|_1$ and $q^{*}_b\leq q^{*}_a$.
\end{prop}
\begin{proof}
%which according to \eqref{Nash eq} results in $ \tilde{v}_k <  \tilde{v}_l$ .
Suppose that $q^*_a < q^*_b$. From Lemma \ref{lem3} we have $\|\vec{S}^*_a\|_1 \leq \|\vec{S}^*_b\|_1$, which contradicts with $K_b \leq K_a$.
%If $K_a\leq K_b$, then from Lemma \ref{lem2} we have $ \tilde{v}_l \leq  \tilde{v}_k$ which from \eqref{Nash eq} implies $q^{*}_a\leq q^{*}_b$.
%If $q_a^{*} < q_b^{*}$, then from Lemma \ref{lem3} we have $\|\vec{S}_a^{*}\|_1 \leq \|\vec{S}_b^{*}\|_1$. Hence, from the budget constraint we have $K_a < K_b$.
%This means that if $K_a\leq K_b$, then $q^{*}_a\leq q^{*}_b$.
Also, suppose $\|\vec{S}_a^{*}\|_1 < \|\vec{S}_b^{*}\|_1$, then from Lemma \ref{lem3} we have $q^{*}_a \leq q^{*}_b$, which contradicts with $K_b \leq K_a$.
%Therefore, from the budget constraint we have $K_a < K_b$.
This completes the proof.
\end{proof}

In the next proposition we explain how the seeding budget and quality at the Nash equilibrium vary with $K_a$ and $K_b$.

\begin{prop} \label{prop quality nash}
Given a fixed graph, the optimal seeding $\|S_a^*\|_1$ and quality $q_a^*$ at the Nash equilibrium are increasing functions of $K_a$. Furthermore, $\|S_a^*\|_1$ is a decreasing function of $K_b$ if $K_b \leq K_a$ and an increasing function of $K_b$ if $K_a \leq K_b$.
\end{prop}
\begin{proof}
%First note that $\|S_a\|_1$, $\|S_b\|_1$, $q_a$, $q_b$ (and as a result $\tilde{v}_k$ and $\tilde{v}_l$) are continues functions of $K_a$ and $K_b$.
First note that $\|S_a^*\|_1$, $\|S_b^*\|_1$, $q_a^*$, $q_b^*$ (and as a result $\tilde{v}_k$ and $\tilde{v}_l$) are continuous functions of $K_a$ and $K_b$. To see this, let $B(q_a,q_b,K_a,K_b)$ denote the best response of the firms to qualities $(q_a,q_b)$ when the budgets are $(K_a,K_b)$. It follows from the continuity of the best response and compactness of action spaces that the set $\{(q_a^*,q_b^*)|B(q_a^*,q_b^*,K_a,K_b)=(q_a^*,q_b^*)\}$, that is the equilibrium space, is closed. This implies that the graphs of the functions $q_a^*(K_a,K_b)$ and $q_b^*(K_a,K_b)$ are closed and thus are continuous.

Now, if $0 < S_{a_k}^{*} < \frac{1}{2}$, then $\tilde{v}_k = v_k$. If $K_a$ marginally increases, then, using the continuity of the equilibrium, the level $k$ and as a result $\tilde{v}_k$ does not change. Thus, given fixed $K_b$, the constraint for firm $b$ in \eqref{last agent} and hence
$\tilde{v}_l$ does not change. Therefore, if $K_a$ marginally increases, from the Nash equilibrium in \eqref{Nash eq}, $q^{*}_a$ does not change and $S_{a_k}^{*}$ marginally increases.
If $S_{a_k}^{*} = 0$ and $v_k < \tilde{v}_k < v_{k-1}$, and $K_a$ marginally increases, from the continuity of Nash we still have $v_k < \tilde{v}_k' < v_{k-1}$ and as a result $S_{a_k}^{*} = 0$ does not change and hence, $q^{*}_a$ marginally increases. If $S_{a_k}^{*} = 0$ and $\tilde{v}_k=v_k$ or $\tilde{v}_k = v_{k-1}$, and $K_a$ marginally increases, either we have $v_k < \tilde{v}_k' < v_{k-1}$ which means $S_{a_k}^{*} = 0$ does not change and $q^{*}_a$ marginally increases, or $\tilde{v}_k$ does not change. In this latter case, given the fixed budget $K_b$,
the constraint for firm $b$ in \eqref{last agent} will remain unchanged and hence $\tilde{v}_l$ will not change.
Therefore, from the Nash equilibrium in \eqref{Nash eq}, $q^{*}_a$ does not change and as a result $S_{a_k}^{*}$ marginally increases.
It is to be noted here that the cases where $\tilde{v}_k$ moves above $v_{k-1}$ or below $v_{k}$ are not feasible as they will cause a jump in the seeding budget, contradicting the continuity of equilibrium. The analysis for the case when $S_{a_n}^{*} = \frac{1}{2}$ and $\tilde{v}_n \leq v_n$ is quite similar.
%Similar argument holds for when $S_{a_k}^{*} = 0$ and $\tilde{v}_k = v_{k-1}$.
Therefore, $\frac{\partial \|\vec{S}_a^{*}\|_1}{\partial K_a}\geq 0$ and $\frac{\partial q_a^{*}}{\partial K_a}\geq 0$.

For the second part of the proposition, if $S_{a_k}^{*}=0$ and $v_k < \tilde{v}_k < v_{k-1}$ and $K_b$ marginally increases, from continuity of $\tilde{v}_k$ we still have $v_k < \tilde{v}_k' < v_{k-1}$, and therefore, $S_{a_k}^{*}=0$ and given the fixed $K_a$, $q_{a}^{*}$ does not change. Hence, we only need to consider the case where $0 < S_{a_k}^{*} < \frac{1}{2}$ and $\tilde{v}_k = v_k$, or $S_{a_k}^{*}=0$ and $\tilde{v}_k = v_k$  or $\tilde{v}_k = v_{k-1}$. In these cases, it is easy to see that either $\tilde{v}_k$ or $S_{a_k}^{*}$ remains unchanged. In the latter case, (given the fixed $K_a$) $q_a^{*}$ does not change. Similar argument holds for when $S_{a_n}^{*} = \frac{1}{2}$ and $\tilde{v}_n \leq v_n$. Therefore, we only need to consider the case where $\tilde{v}_k$ does not change.
From the first part of the proposition, $q_b^{*}$ is an increasing function of $K_b$. Also, from Proposition \ref{prop compare nash}, if $K_b \leq K_a$ ($K_a \leq K_b$), then $q_b^{*} \leq q_a^{*}$ ($q_a^{*} \leq q_b^{*}$). Therefore, if $K_b \leq K_a$ ($K_a \leq K_b$) and $K_b$ marginally increases, equations \eqref{tilde} implies that $q_a^{*}$ must marginally increase (decrease) or does not change so that $\tilde{v}_k$ remains fixed. Hence, given constant $K_a$, $S_{a_k}^{*}$ marginally decreases (increases) or does not change.
Therefore, $\frac{\partial \|\vec{S}_a^{*}\|_1}{\partial K_b} \leq 0,$ for $K_b \leq K_a$ and $\frac{\partial \|\vec{S}_a^{*}\|_1}{\partial K_b} \geq 0,$ for $K_a \leq K_b$.

\end{proof}

Proposition \ref{prop quality nash} implies that when $K_b \leq K_a$, the higher the budget of the rival firm, the lower the seeding budget of firm $a$, i.e., if $K_b \leq K_b'\leq K_a$ then, $\|\vec{S}_a^*(K_b')\|_1\leq\|\vec{S}_a^*(K_b)\|_1$. On the other hand, when competing with a firm which has a higher budget, i.e. $K_a \leq K_b$, the higher the budget of the rival firm, the higher firm $a$ should spend on seeding. In other words, if $K_a\leq K_b\leq K_b'$ then, $\|\vec{S}_a^*(K_b)\|_1\leq\|\vec{S}_a^*(K_b')\|_1$.

Combining these two results, we can see that given a fixed value of $K_a$, the seeding budget of firm $a$ is increasing with the difference $|K_a-K_b|$. The seeding budget attains its minimum when $K_b=K_a$, implying that the firm should allocate more budget to quality to distance itself from the rival firm. However, as the gap between budget widens, competition in qualities becomes less effective and firms spend more budget on seeding.
%Proposition \ref{prop quality nash} demonstrates the effect of each budget on the optimal seeding budget when the other budget is fixed.

\subsection{Effect of Network Structure on Firms' Decisions}

%\item If there is no budget constraint, then the maximum ratio of qualities of firms happens in a graph with one agent with centrality of $v_h^{s}$ and another agent with centrality of one. An example of such graphs is a directed star graph where all edges are directed towards the center except one edge which goes both ways. The minimum happens in a graph with two agents with equal centralities. An example of such graphs is a balanced graph.

%\item Given specific equal budgets $K_a=K_b=K$, then the minimum quality and maximum seeding occurs in a graph which satisfies
%\begin{equation*}
%0 < K - \frac{l-1}{2}-\frac{\lambda} {2\tilde{v}_l} \leq \frac{1}{2}.
%\end{equation*}
%where $\tilde{v}_l$ is the maximum value in $[v^{max}_{l+1},  v^{max}_{l}]$ which satisfies the above equation, and $v^{max}_{l} = \max v_l$ in all graphs. A graph in which $v_l = v^{max}_{l}$ is a $l$ star graph which has the minimum quality and maximum seeding.

In this subsection we study the effect of network structure on the Nash equilibrium. Since we already studied the effect of the budget on the Nash equilibrium, for the rest of this subsection we assume $K_a = K_b = K$ so that we can observe only the effect of the network structure. We first focus on two well studied graphs, i.e. star and balanced graphs, and highlight how they can reflect important properties of the seeding budget. Before continuing further, we first formally define these two graphs and find their network centralities in the next lemma.

\begin{mydef}
A star graph is a directed graph in which there is an edge from any noncentral agent $i \in V - \{ j \}$ to the central agent $j$ with the weight $g_{ij}=1$ and there are edges from the central agent $j$ to all noncentral agents $i \in V - \{ j \}$ such that $\sum_{i} g_{ji} = 1$.
\end{mydef}

\begin{mydef}
A balanced graph is a directed graph in which the in-degree of each agent is equal to its out-degree, i.e. $\sum_j g_{ji} = \sum_j g_{ij} = 1$.
\end{mydef}

\begin{lem}\label{v_Lv_H}
The centrality of the agents in a balanced graph is given by $\bar{v}=\frac{2\beta}{2\beta-\delta}$. In a star graph, the centrality of the central agent is
\begin{equation*}
v_h^{s}=\frac{1+\frac{\delta(n-1)}{2\beta}}{1-(\frac{\delta}{2\beta})^2},
\end{equation*}
and the centrality of non central agents is
\begin{equation*}
v_l^{s}=\frac{1+\frac{\delta}{2\beta(n-1)}}{1-(\frac{\delta}{2\beta})^2}.
\end{equation*}
Moreover, for any arbitrary graph $G$, $\bar{v}\leq v_1 \leq v_h^{s}$.
%, where $v_{max}=\max_{i\in V}v_i$.
\end{lem}
\begin{proof}
First part simply follows from the fact that $v=(I-\delta W^T)^{-1}\vec{1}$, where $W$ is given by \eqref{definition}, and that for any agent $i$ in a balanced graph $\sum g_{ji}=\sum g_{ij}=1$. For the star graph, noting that $v=\vec{1}+\delta W^T v$, we can obtain
\begin{align*}
v_h^{s}&=1+ \frac{\delta(n-1)v_l^{s}}{2\beta},\\
v_l^{s}&=1+ \frac{\delta v_h^{s}}{2\beta(n-1)},
\end{align*}
solving which we can find $v_h^{s}$ and $v_l^{s}$ as given in the lemma.

Also, for any arbitrary graph $G$, $v_1 \geq\frac{\sum v_i}{n}=\bar{v}$. To show $v_1 \leq v_h^{s}$,
%assume that the maximum centrality occurs for agent $i$. Again,
using $v=\vec{1}+\delta W^T v$ for all $j\neq 1$ we can obtain
\begin{equation*}
v_j\geq 1+ (\frac{\delta}{2\beta}) g_{1j} v_1.
\end{equation*}
This yields
\begin{equation*}
\sum_{j=1}^{n}{v_j}\geq (n-1)+(1+\frac{\delta}{2\beta})v_1.
\end{equation*}
Applying simple algebra along with the fact that $\sum v_j=\frac{2\beta n}{2\beta-\delta}$ leads to $v_1\leq v_h^{s}$.
\end{proof}

The next proposition provides a condition for seeding profitability of any general graph. Also, the seeding budget of star and balanced graphs are compared and it is shown that the graph with higher seeding budget can be any of the two, depending on the budget.

\begin{prop} \label{compare seeding nash}
If seeding budget is nonzero for a balanced graph, it will be nonzero for any other graph too. On the other hand, if seeding budget is zero for a star graph, it will also be zero for any other graph.
Moreover, if $\frac{1}{2} + \frac{\lambda}{2 \bar{v}} < \frac{K}{c_s} < \frac{n}{2} + \frac{\lambda}{2 v_l^s}$, a balanced graph has a larger seeding budget than a star graph, and if $\frac{\lambda}{2 v_h^s} < \frac{K}{c_s} < \frac{1}{2} + \frac{\lambda}{2 \bar{v}}$, a star graph has a larger seeding budget than a balanced graph. For  $\frac{n}{2} + \frac{\lambda}{2 v_l^s} < \frac{K}{c_s} < \frac{n}{2} + \frac{\lambda}{2}$ they have the same seeding budget.
\end{prop}
\begin{proof}
If seeding budget is nonzero for a balanced graph, then according to \eqref{last agent symmetric nash} we have $\frac{\lambda} {2\bar{v}} < \frac{K}{c_s}$. As a result, for any other graph we will have $\frac{\lambda} {2v_1} < \frac{K}{c_s}$, since according to Lemma~\ref{v_Lv_H} $\bar{v}\leq v_1$.
This means that there exists at least one agent that must be seeded. On the other hand, if seeding budget is zero for a star graph, then we must have $\frac{K}{c_s} \leq \frac{\lambda} {2 v_h^s}$. Since we know $v_h^s \geq v_i$ for any agent $i$ of any arbitrary graph, therefore, $\frac{K}{c_s}\leq \frac{\lambda} {2 v_i}$ and no agent can be seeded in any other graph.

For the second part of the proposition, denote quality and seeding budget of balanced and star graphs by
$q_r$, $\|\vec{S}_r\|_1$ and $q_s$, $\|\vec{S}_s\|_1$ respectively. If $\frac{1}{2} + \frac{\lambda}{2 \bar{v}} < \frac{K}{c_s} < \frac{n}{2} + \frac{\lambda}{2 v_l^s}$, then seeding budget is nonzero for balanced graph and hence, $q_r = (\frac{c_s}{c_q})\frac{\lambda}{2 \bar{v}}$. This implies $\|\vec{S}_r\|_1 = \frac{K}{c_s} - \frac{\lambda}{2 \bar{v}} > \frac{1}{2}$. For star graph $\frac{1}{2} + \frac{\lambda}{2 \bar{v}} < \frac{K}{c_s} < \frac{n}{2} + \frac{\lambda}{2 v_l^s}$ implies $\frac{1}{2} + \frac{\lambda}{2 v_h^s} < \frac{K}{c_s}$. Therefore, the central agent in star graph must be seeded, i.e. $S_{a_1}=\frac{1}{2}$. If $S_{a_2}=0$, then $\| \vec{S}_s\|_1 = \frac{1}{2}$ and clearly $\| \vec{S}_s\|_1 < \| \vec{S}_r\|_1$. If $S_{a_2} > 0$, then a noncentral agent must be seeded and we must have $q_s = (\frac{c_s}{c_q})\frac{\lambda}{2 v_l^s}$. This implies $q_r < q_s$ and as a result $\| \vec{S}_s\|_1 < \| \vec{S}_r\|_1$.

If $ \frac{\lambda}{2 v_h^s} <\frac{K}{c_s} < \frac{1}{2} + \frac{\lambda}{2 \bar{v}}$, then we have two cases. If $\frac{K}{c_s} \leq \frac{\lambda}{2 \bar{v}}$ then $\|\vec{S}_r\|_1 = 0$. On the other hand $\|\vec{S}_s\|_1 > 0$ since $ \frac{\lambda}{2 v_h^s} <\frac{K}{c_s}$. Therefore, clearly $\| \vec{S}_r\|_1 < \| \vec{S}_s\|_1$. So let's assume $ \frac{\lambda}{2 \bar{v}} < \frac{K}{c_s} < \frac{1}{2} + \frac{\lambda}{2 \bar{v}}$. This implies seeding budget is nonzero for balanced graph and hence, $q_r = (\frac{c_s}{c_q}) \frac{\lambda}{2 \bar{v}}$. As a result, $\| \vec{S}_r\|_1 = \frac{K}{c_s} - \frac{\lambda}{2 \bar{v}}< \frac{1}{2}$. Now again consider two cases. If
$\frac{K}{c_s} < \frac{1}{2} + \frac{\lambda}{2 v_h^s}$, then $q_s = (\frac{c_s}{c_q}) \frac{\lambda}{2 v_h^s}$, and hence $q_s < q_r$. This implies $\| \vec{S}_r\|_1 < \| \vec{S}_s\|_1$. Otherwise, if $\frac{1}{2} + \frac{\lambda}{2 v_h^s} \leq \frac{K}{c_s}$ then $\| \vec{S}_s\|_1 \geq \frac{1}{2}$. As a result, again we have $\| \vec{S}_r\|_1 < \frac{1}{2} \leq \| \vec{S}_s\|_1$.

If $\frac{n}{2} + \frac{\lambda}{2 v_l^s} < \frac{K}{c_s} < \frac{n}{2} + \frac{\lambda}{2}$, then all agents in star graph are seeded up to agents maximum demand capacities which is $0.5$ for each agent. Also, since $v_l^s < \bar{v}$, we have $\frac{n}{2} + \frac{\lambda}{2 \bar{v} } < \frac{K}{c_s}$. Hence, all agents in balanced graph are also seeded up to agents maximum demand capacities. Therefore, both graphs have the same seeding budget. This completes the proof.
 %
% $\frac{K}{c_s} < \frac{1}{2} + \frac{\lambda}{2 \bar{v}} < \frac{1}{2} + \frac{\lambda}{2 v_l^s}$ and
\end{proof}
%
%0 < \frac{K}{c_s} - \frac{l-1}{2}-\frac{\lambda} {2\tilde{v}_l} \leq \frac{1}{2}. \\

The next proposition provides us with a lower and an upper bound for minimum and maximum seeding budget.
In order to characterize the graphs with maximum and minimum seedings for a given budget $K$, we need to introduce a few notations first.
\begin{mydef}
Define $v_l^{max} = \max v_l$, i.e. the maximum of the $l$-th centrality $v_l$ among all possible graphs subject to $\sum v_i = \frac{2\beta n}{2\beta- \delta}$. We can see that $v^{max}_{1} = v_h^s$ and
\begin{equation} \label{l star graph}
v^{max}_{l} = \frac{n\delta}{l(2\beta - \delta)} + 1,
\end{equation}
for $l \geq 2$. Similarly, define $v_l^{min} = \min v_l$, i.e. the minimum of the $l$-th centrality $v_l$ among all possible graphs subject to $\sum v_i = \frac{2\beta n}{2\beta- \delta}$. It is easy to see that
$v^{min}_{1} = \bar{v}$, $v^{min}_{2} = v_l^s$, and $v^{min}_{l} = 1$ for $l \geq 3$.
\end{mydef}

\begin{prop} \label{max seeding nash}
Let $(l,\tilde{v}_l^{max})$ be the unique pair satisfying condition \eqref{last agent symmetric nash} where $v_{l}^{max}  \leq \tilde{v}_l^{max} \leq v_{l-1}^{max}$ and if $0 < S^{*}_{a_l}$ in \eqref{last agent symmetric nash}, then $\tilde{v}_l^{max}=v_l^{max}$.
The maximum seeding budget occurs in any graph for which $\tilde v_l = \tilde{v}_{l}^{max}$. An example for such a graph is an $l$-star graph if $\tilde{v}_l^{max}=v_{l}^{max}$ and an $(l-1)$-star graph if $\tilde{v}_l^{max}>v_{l}^{max}$.
\footnote{$\tilde{v}^{max}_n \leq v^{max}_{n}$ if $S^{*}_{a_n}= \frac{1}{2}$.}
Similarly, let $(l,\tilde{v}_l^{min})$ be the unique pair satisfying condition \eqref{last agent symmetric nash} where $v_{l}^{min}  \leq \tilde{v}_l^{min} \leq v_{l-1}^{min}$ and if $0 < S^{*}_{a_l}$ in \eqref{last agent symmetric nash}, then $\tilde{v}_l^{min}=v_l^{min}$. The minimum seeding budget occurs in any graph for which $\tilde v_l = \tilde{v}_{l}^{min}$. An example for such graphs is the balanced graph for $l = 1$, the star graph for $l = 2$, and any graph with $n-2$ agents with centrality of one for $l\geq3$.
\footnote{$\tilde{v}_n^{min} \leq v^{min}_{n}$ if $S^{*}_{a_n}= \frac{1}{2}$.}
\end{prop}

\begin{proof}
Let $G$ be a graph attaining the maximum seeding (thus the minimum quality) and denote its corresponding equilibrium with $(l',\tilde v_{l'})$. Note that $l'\geq l$, since in a graph with $\tilde v_l=\tilde v_l^{max}$ the first $(l-1)$ agents are fully seeded. Now, if $l'>l$, then from $\tilde v_{l'}\leq v_{l'-1}^{max}\leq v_{l}^{max}$ and $v_{l}^{max}  \leq \tilde{v}_l^{max}$ it follows that $\tilde v_{l'}\leq\tilde{v}_l^{max}$. But, then both pairs $(l',\tilde v_{l'})$ and $(l,\tilde{v}_l^{max})$ cannot satisfy \eqref{last agent symmetric nash}. Therefore, in a graph with maximum seeding we should have $l'=l$. Now, if $\tilde v_l^{max}<\tilde v_{l'}$, then $v_l^{max}<\tilde v_{l'}\leq v_{l-1}^{max}$, which contradicts the uniqueness of the pair $(l,\tilde v_l^{max})$. To complete the proof, we also need to show that $\tilde v_l=\tilde v_l^{max}$ is achievable. It is quite straightforward to show that for $\tilde{v}_l^{max}=v_{l}^{max}$ an $l$-star graph with $v_1=\ldots=v_l=v_l^{max}$, and for $\tilde{v}_l^{max}>v_{l}^{max}$ an $(l-1)$-star graph with $v_1=\ldots=v_{l-1}=v_{l-1}^{max}$ admit $(l,\tilde v_l^{max})$ as the equilibrium. The proof for the minimum seeding budget is similar.
\end{proof}

\begin{ex} \label{ex1}
As a numerical example for the minimum and maximum seeding budgets, we consider a network with $n=15$ agents with budget $K=2$, quality and seeding costs of $c_s=c_q=1$ and parameters of $\alpha = \beta = 1$ and $\delta=0.5$. For this example from equations \eqref{lambda} we have $\lambda= 5$. As a result, we can see that for $l=3$ and $v^{max}_{3}=\frac{8}{3}$ from \eqref{l star graph}, condition $0 < S^{*}_{a_3}=\frac{1}{16} < \frac{1}{2}$ in \eqref{last agent symmetric nash} is satisfied. Therefore, a graph with the maximum seeding budget is a $3$-star with the seeding budget of $\frac{17}{16}$ as illustrated in Fig.~\ref{fig:3star}.
Also, we can see that for $l=1$ and $\bar v=\frac{4}{3}$, condition $0 < S^{*}_{a_1}=\frac{1}{8} < \frac{1}{2}$ in \eqref{last agent symmetric nash} is satisfied. Thus, balanced graph has the minimum seeding budget of $\frac{1}{8}$. Given $v_h^s = 4.8$ and $v_l^s = 1.08$, it can be seen that in star graph $\tilde{v}_2=\frac{5}{3}$ and star graph has a seeding budget of $0.5$ which is neither a minimum nor a maximum.
\end{ex}
\begin{figure}
\centering
\includegraphics[scale=2]{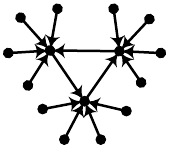}
\caption{A graph with maximum seeding budget}\label{fig:3star}
\end{figure}

As we saw, the structure of the graphs with minimum and maximum seeding budget depends on the budget. However, for certain values of budget $K$ the seeding budget will be independent of the structure of the graph, as described in the next proposition.
\begin{prop} \label{indep of network nash}
If $\frac{K}{c_s} < \frac{\lambda}{2v_h^{s}}$ no graph can be seeded. On the other hand, if $\frac{K}{c_s} > \frac{n}{2} + \frac{\lambda}{2}$ all graphs can be seeded up to agents maximum demand capacities.
\end{prop}
\begin{proof}
The maximum possible centrality happens for the central agent of the star graph as shown in Lemma~\ref{v_Lv_H}. As a result, if $\frac{K}{c_s} < \frac{\lambda}{2v_h^{s}}$, then we have $\frac{K}{c_s} < \frac{\lambda}{2v_i}$ for all $i$ in any graph and from condition \eqref{last agent symmetric nash} no agent can be seeded. Also, since from definition $1 \leq v_i$ for all $i$, if $\frac{K}{c_s} > \frac{1}{2} + \frac{\lambda}{2}$ then we have $\frac{K}{c_s} > \frac{n}{2} + \frac{\lambda}{2v_i}$ for all agents in any graph, and any graph can be seeded up to agents' maximum demand capacities which is $0.5$ for each agent.
\end{proof}

\section{Seeding Versus Quality Improvement}  \label{sec4}

In this section we describe a scenario in which firms have already produced their products with some preset quality. We assume at some point in time, say $t=0$, firms learn about the network structure and utilize a fixed budget to maximize their marginal utility by either marginally ``improving the quality'' of their products or new seeding some agents to change their consumption towards their products or both. Since the products have been in the market for a while, we assume agents have already decided on their consumption from products $a$ and $b$ which are denoted by $\vec{x}(0)$ and $\vec{1}-\vec{x}(0)$ respectively. Each firm has a limited budget, i.e. $K_a$ and $K_b$, that can spend on either new seeding, i.e. $\vec{S}_a$ and $\vec{S}_b$, or enhancing the quality of its product, i.e. $\Delta q_a$ and $\Delta q_b$, or both. New seeding $\vec{S}_a$ and $\vec{S}_b$ will change the initial consumption of products $a$ and $b$ by $\vec{S}_a-\vec{S}_b$ and $\vec{S}_b-\vec{S}_a$ respectively. In order to have $0 \leq x_i(0) \leq 1$ and $0 \leq 1 - x_i(0) \leq 1$ for all agents $i$, we impose the constraints $\|\vec{S}_a+\vec{y}(0)\|_\infty\leq 0.5$ and $\|\vec{S}_b-\vec{y}(0)\|_\infty\leq 0.5$. This means that firms can initially seed agents up to their demand capacity. From equation \eqref{payoff} the marginal change in the utility of firm $a$ and $b$ resulted from the new budget $K_a$ and $K_b$ are given by
\begin{equation*}
\begin{split}
&\Delta U_a=v^T\vec{S}_a - v^T\vec{S}_b + \frac{2\lambda q_b\Delta q_a}{(q_a+q_b)^2} - \frac{2\lambda q_a\Delta q_b}{(q_a+q_b)^2}, \\
&\Delta U_b=v^T\vec{S}_b - v^T\vec{S}_a + \frac{2\lambda q_a\Delta q_b}{(q_a+q_b)^2} - \frac{2\lambda q_b\Delta q_a}{(q_a+q_b)^2}.
\end{split}
\end{equation*}
%We consider a game between firms where each firm maximizes its product consumption given a fixed budget.
%As mentioned earlier, the utility of each firm is the discounted sum of its product consumption over time as described by \eqref{payoff}.
%\begin{equation*}
%\begin{split}
%&U_a(q_a,q_b) =  \sum_{t=0}^{\infty} \delta^t \vec{1}^T (\frac{\vec{1}}{2} + \vec{y}(t)), \\
%&U_b(q_a,q_b) =  \sum_{t=0}^{\infty} \delta^t \vec{1}^T (\frac{\vec{1}}{2} - \vec{y}(t)).
%\end{split}
%\end{equation*}
%where $c_a$ and $c_b$ are costs of producing each unit of products $a$ and $b$ respectively.
%Since $(p - c)$ is a constant for both firms we can just maximize the discounted product consumption of products over time, i.e. $\tilde{U}(q_a,q_b)$ where $U(q_a,q_b) = (p - c) \tilde{U}(q_a,q_b)$, instead of the discounted profit over time i.e. $U(q_a,q_b)$.
%Using equations \eqref{definition} and \eqref{expanded form} and defining the centrality vector $v$ by $v= (I-\delta W^T)^{-1} \vec{1}$
%and noting that $\sum v_i = \frac{2\beta n}{2\beta- \delta}$,
%the utilities of firms can be written as
%\begin{equation} \label{payoffs}
%\begin{split}
%& U_a(q_a,q_b) = (\frac{n}{2(1-\delta)}) + v^T \vec{y}(0) + \lambda (\frac{q_a-q_b}{q_a+q_b}), \\
%& U_b(q_a,q_b) = (\frac{n}{2(1-\delta)}) - v^T \vec{y}(0) - \lambda (\frac{q_a-q_b}{q_a+q_b}),
%\end{split}
%\end{equation}
%where
%\begin{equation} \label{lambda}
%\lambda = \frac{\delta(1+2(\alpha-\beta)) n}{2(1-\delta)(2\beta- \delta)}.
%\end{equation}
%As it can be seen from equations above, firms play a fixed-sum game.
We assume the cost of improving the quality by $\Delta q$ is given by $c_q \Delta q$ and $c_q$ is a large number, and also the cost of each unit of new seeding is given by $c_s$. Each firm maximizes its marginal utility given its fixed budget. Since the effect of the action of firm $b$, i.e. $\vec{S}_b$ and $\Delta q_b$, is decoupled from that of the action of firm $a$ in $\Delta U_a$, thus firm $a$ should solve the following optimization problem
\begin{equation}  \label{firms utility function a}
\begin{split}
& \max_{\vec{S}_a, \Delta q_a} \quad v^T\vec{S}_a+\frac{2\lambda q_b\Delta q_a}{(q_a+q_b)^2}, \\
& \text{s.t.} \quad c_s \|\vec{S}_a\|_1 + c_q \Delta q_a=K_a.
\end{split}
\end{equation}
Similarly, for the firm $b$ we have
\begin{equation}  \label{firms utility function b}
\begin{split}
& \max_{\vec{S}_b, \Delta q_b} \quad v^T\vec{S}_b+\frac{2\lambda q_a\Delta q_b}{(q_a+q_b)^2}, \\
& \text{s.t.} \quad c_s \|\vec{S}_b\|_1 + c_q \Delta q_b=K_b.
\end{split}
\end{equation}
From equations \eqref{firms utility function a} and \eqref{firms utility function b} it can be seen that the optimal strategy of each firm is independent of the action of the other firm. It is to be noted that despite the independence of the actions, the optimal strategy of each firm depends on the state (i.e., quality) of the rival firm. This results in a Nash equilibrium to be simply the pair of the optimal actions of the firms. In the next Theorem we describe a simple rule for the optimal allocation of the budget for each firm and discuss the resulting Nash equilibrium.

%\section{Optimal Budget Allocation Strategy} \label{sec4}

%In this section we study how firms can maximize the spread of their products by optimally allocating their fixed budget between improving the quality of their products and new seeding.  We also describe a Nash equilibrium of the game played between these two firms described in Section \ref{sec3}.

\begin{thm} \label{thm 2}
For firm $a$, it is more profitable to seed agent $j$ rather than enhancing the quality of its product if $v_j > v_c^a$ where
\begin{equation} \label{condition}
v_c^a \triangleq (2\lambda) (\frac{c_s}{c_q}) (\frac{q_b}{(q_a+q_b)^2}).
\end{equation}
Similarly, for firm $b$, it is more profitable to seed agent $j$ rather than enhancing the quality of its product if $v_j > v_c^b$ where
\begin{equation} \label{condition2}
v_c^b \triangleq (2\lambda) (\frac{c_s}{c_q}) (\frac{q_a}{(q_a+q_b)^2}).
\end{equation}
Moreover, any pair of the optimal strategies of the firms described by the above threshold rules describes a Nash equilibrium.
\end{thm}
\begin{proof}
From equation \eqref{firms utility function a} and \eqref{firms utility function b} the relative marginal utility to cost for spending budget to seed agent $j$ is $\frac{v_j}{c_s}$. Therefore, it is always more profitable to seed an agent with higher centrality. Also, the relative marginal utility to cost for spending budget on enhancing quality of product $a$ is $\frac{2 \lambda q_b}{c_q(q_a+q_b)^2}$ according to \eqref{firms utility function a}. Therefore, for firm $a$ it is more profitable to seed agent $j$ rather than enhancing the quality of its product iff
\begin{equation*}
\frac{v_j}{c_s} > \frac{2 \lambda q_b}{c_q(q_a+q_b)^2}.
\end{equation*}
This completes the proof. Similar story holds for firm $b$. Moreover, since the best response of each firm resulting from equations \eqref{firms utility function a} and \eqref{firms utility function b} is independent of the action of the other firm, any Nash equilibrium of the game between firms is simply a pair of firms best responses.
\end{proof}

\begin{cor}
If firms have equal qualities $q_a=q_b=q$, for both firms $a$ and $b$, it is more profitable to seed agent $j$ rather than enhancing the quality of their products if $v_j > v_c$ where
\begin{equation} \label{condition symmetric}
v_c \triangleq (\frac{\lambda}{2}) (\frac{c_s}{c_q}) (\frac{1}{q}).
\end{equation}
Moreover, any pair of the optimal strategies of the firms described by the above threshold rules describes a Nash equilibrium.
\end{cor}

\begin{rem}
If we compare the thresholds
\begin{equation} \label{threshold}
v_c^a \triangleq (2\lambda) (\frac{c_s}{c_q}) (\frac{q_b}{(q_a+q_b)^2}), \qquad \quad
v_c^b \triangleq (2\lambda) (\frac{c_s}{c_q}) (\frac{q_a}{(q_a+q_b)^2}),
\end{equation}
with qualities in Section \ref{sec3}
\begin{equation} \label{qualities}
q_a^*= (2\lambda)(\frac{c_s}{c_q})(\frac{\tilde{v}_l}{(\tilde{v}_k+\tilde{v}_l)^2}), \qquad \quad
q_b^*= (2\lambda)(\frac{c_s}{c_q}) (\frac{\tilde{v}_k}{(\tilde{v}_k+\tilde{v}_l)^2}),
\end{equation}
we can see a similarity as follows: In equation \eqref{threshold}, $q_a$ and $q_b$ determine $v_c^a$ and $v_c^b$ which in turn determine the trade off between $\vec{S}$ and $\Delta q$ according to Theorem \ref{thm 2}. In Section \ref{sec3}, $K_a$ and $K_b$ determine $\tilde{v}_k$ and $\tilde{v}_l$ based on the inequalities in \eqref{last agent} and $\tilde{v}_k$ and $\tilde{v}_l$ determine $q_a^*$ and $q_a^*$ according to \eqref{qualities}, which in turn determine the trade of between $\vec{S}$ and $q$ based on the budget constraint.
Therefore, as it will be discussed later, we can achieve similar results for the effect of $q_a$ and $q_b$ on the optimal budget allocation, as we did for the effect of $K_a$ and $K_b$ on the Nash equilibrium.
\end{rem}

Following the above theorem, the optimal allocation of the budget for each firm is to follow a so called water-filling strategy, that is, to start seeding in the order of agents' centralities until the centrality falls below the threshold given by \eqref{condition} for firm $a$ or \eqref{condition2} for firm $b$ (in which case the firm spends the rest of the budget on improving the quality), or the firm runs out of budget. Also, the amount that agents can be seeded is up to their demand capacity, i.e. $\vec{S}_a^{max} = (0.5)\vec{1} - \vec{y}(0) > 0$ and $\vec{S}_b^{max} = (0.5)\vec{1} + \vec{y}(0) > 0$. Also, note that if the centrality of any agent is equal to the threshold defined in \eqref{condition} or \eqref{condition2},
then firms are indifferent between seeding that agent and quality improvement. Equations \eqref{condition} and \eqref{condition2} indicate that the optimal allocation depends on quality of products, i.e. $q_a$ and $q_b$ and centrality distribution of agents in the network, i.e. $v$.
%, the parameter $\lambda$ which depends on discounting factor $\delta$ and isolation payoff parameters $\alpha$ and $\beta$, and cost of seeding and quality improvement, i.e. $c_s$ and $c_q$.
In what follows, we will study the effect of each of these factors
%the effect of the network structure and qualities of the products, as well as model parameters,
on the optimal allocation of the firms in more details. All of our analysis here is for firm $a$ and similar results can be shown for firm $b$ as well. For simplicity,
%when the parameters of the budget constraint $c_s \|\vec{S}_a\|_1 + c_q \Delta q_a=K_a$ are fixed, i.e. $c_s$, $c_q$ and $K_a$ are fixed,
we only discuss optimal seeding budget; optimal quality improvement budget can be found easily using the budget constraint.

%For simplicity, we only discuss optimal seeding budget; optimal quality improvement budget can be found easily using the budget constraint $c_s \|\vec{S}_a\|_1 + c_q \Delta q_a=K_a$. All of our analysis here is for firm $a$ and similar results can be shown for firm $b$ as well.

\subsection{Effect of Quality of Products on Firms' Decisions:}
In this subsection we study how the quality of each product, i.e. $q_a$ and $q_b$, can influence the optimal allocation of seeding and quality improvement budgets.

As it can be seen from equation \eqref{condition}, the threshold $v_c^{a}$ depends on both firm's and its rival's qualities, i.e. both $q_a$ and $q_b$. In the next proposition, we compare the seeding budget of two firms in the optimal allocation with respect to their qualities.
\begin{prop} \label{prop compare allocation}
Given an equal budget, the firm with higher quality also has higher seeding budget, i.e. if $q_a\leq q_b$, then $\|\vec{S}_a^*\|_1\leq\|\vec{S}_b^*\|_1$.
\end{prop}
\begin{proof}
From equations \eqref{condition} and \eqref{condition2} it can be easily seen that if $q_a\leq q_b$, then $v_c^{b}\leq v_c^{a}$. As a result, more agents satisfy the condition \eqref{condition2} for firm $b$ compared to firm $a$ and therefore, $\|\vec{S}_a^*\|_1\leq\|\vec{S}_b^*\|_1$.
\end{proof}

This result is due to diminishing return of quality which means if a firm already has a high quality it would profit less by spending on quality improvement and it would be better for the firm to invest on seeding. Also, note that the result of Proposition \ref{prop compare allocation} is similar to the result of Proposition \ref{prop compare nash}. The only difference is that instead of budgets $K_a$ and $K_b$ in Proposition \ref{prop compare nash}, qualities $q_a$ and $q_b$ in Proposition \ref{prop compare allocation} play the role of the budgets while comparing the seedings of the firms.

In the next proposition we explain how the optimal seeding budget vary with $q_a$ and $q_b$.
\begin{prop} \label{prop quality}
Given a fixed graph, the optimal seeding budget is an increasing function of $q_a$. Furthermore, it is a decreasing function of $q_b$ if $q_b \leq q_a$ and an increasing function of $q_b$ if $q_a \leq q_b$.
\end{prop}
\begin{proof}
The optimal seeding budget is a decreasing function of the threshold value $v_c^{a}$. %
%because if the threshold value $v_c^{a}$ increases then there will be less (or same) agents with centrality above the threshold which can be seeded in the optimal allocation. Therefore, the optimal seeding budget decreases.
Also, it is easy to see that the threshold value of $v_c^{a}$ is a decreasing function of $q_a$. This implies the first part of proposition. For the second part, it is easy to see that
%$$ \frac{\partial v_c^{a}}{\partial q_a} = k \times \frac{-q_b}{(q_a+q_b)^3} <0,$$ for some constant $k > 0$. Therefore,
the threshold value of $v_c^{a}$ is a decreasing function of $q_a$. This implies the first part of proposition. For the second part, it is easy to see that
%we have
%$$ \frac{\partial v_c^{a}}{\partial q_b} = k \times \frac{(q_a-q_b)}{(q_a+q_b)^3},$$ for some constant $k > 0$. Therefore,
the threshold value of $v_c^{a}$ is an increasing function of the quality of product $b$, when $q_b \leq q_a$ and a decreasing function of the quality of product $b$, when $q_b \geq q_a$. This completes the proof.
\end{proof}

Proposition \ref{prop quality} implies that a higher quality in a firm's product results in a higher seeding budget in the optimal allocation. This can be due to the diminishing return of quality: when quality is higher there is less need for quality improvement and it would be more profitable to spend on seeding. Furthermore, when $q_b \leq q_a$, the higher the quality of the rival firm's product, the lower the seeding budget of firm $a$, i.e., if $q_a\geq q_b'\geq q_b$ then, $\|\vec{S}_a^*(q_b')\|_1\leq\|\vec{S}_a^*(q_b)\|_1$. On the other hand, when competing with a firm whose product has a higher quality, i.e. $q_b \geq q_a$, the higher the quality of the rival firm's product, the higher firm $a$ should spend on seeding. In other words, if $q_a\leq q_b\leq q_b'$ then, $\|\vec{S}_a^*(q_b)\|_1\leq\|\vec{S}_a^*(q_b')\|_1$.

Combining these two results, we can see that given a fixed value of $q_a$, the seeding budget of firm $a$ is increasing with the difference $|q_a-q_b|$. The seeding budget attains its minimum when $q_b=q_a$, implying that the firm should allocate more budget to quality improvement to distance itself from the rival firm. However, as the gap between qualities widens, competition in qualities becomes less effective and firms spend more budget on seeding.
%\begin{rem}
Also, note that the result of Proposition \ref{prop quality} is similar to the result of Proposition \ref{prop quality nash}. The only difference is that seeding budgets vary with $q_a$ and $q_b$ in Proposition \ref{prop quality}, whereas they vary with $K_a$ and $K_b$ in Proposition \ref{prop quality nash}.
%\end{rem}

%Proposition \ref{prop quality} demonstrates the effect of each quality on the optimal seeding budget when the other quality is fixed.

\subsection{Effect of Network Structure on Firms' Decisions:}

In this subsection we study the effect of network structure on the optimal allocation of the budget for seeding and quality improvement. First we define seeding capacity of a graph.
\begin{mydef}
The seeding capacity of a graph is the amount that it can be seeded in the optimal allocation when there is no budget constraint.
\end{mydef}

We first focus on two well studied graphs, i.e. star and balanced graphs, and highlight how they can reflect  important properties of the seeding budget.
The next proposition provides a condition for seeding profitability of any general graph. Also, the seeding capacity of star and balanced graphs are compared and it is shown that the graph with higher seeding capacity can be any of the two, depending on the threshold value of $v_c^{a}$ in \eqref{condition}.
\begin{prop} \label{compare seeding}
If seeding capacity is nonzero for a balanced graph, it will be nonzero for any other graph too. On the other hand, if seeding capacity is zero for a star graph, it will also be zero for any other graph. Moreover, if $ v_l^{s} < v_c^a < \bar{v}$, a balanced graph has a larger seeding capacity than a star graph, and if $\bar{v}< v_c^a <  v_h^{s}$, a star graph has a larger seeding capacity than a balanced graph. For $1<v_c^a<v_l^{s}$ they have the same seeding capacity.
\end{prop}
\begin{proof}
If seeding capacity is nonzero for a balanced graph, then we have  $v_c^a < \bar{v}$. As a result, for any other graph we will have $v_c^a <v_{max}$, where $v_{max}=\max v_i$, since according to Lemma \ref{v_Lv_H} $\bar{v}\leq v_{max}$. This means that there exists at least one agent that must be seeded. On the other hand, if seeding capacity is zero for a star graph, then we must have $v_c^a > v_h^{s}$. Since we know $v_h^{s} \geq v_i$ for any agent $i$ of any arbitrary graph, therefore, $v_c^a > v_i$ and no agent can be seeded in any other graph.
For the second part of the proposition,
if $ v_l^{s} < v_c^a < \bar{v}$, then seeding capacity for the star graph will be $S^{max}_{a_i}$, where $S^{max}_{a_1} \geq S^{max}_{a_2} \geq \cdots \geq S^{max}_{a_n}$ are elements of the demand capacity vector $\vec{S}^{max}_a$ and agent $i$ is the central agent. However, for the balanced graph all agents can be seeded up to their maximum demand capacities and the seeding capacity will be $\|\vec{S}^{max}_a\|_1$. On the other hand, if $ \bar v < v_c^a < v_h^{s}$, still seeding for the star graph will be $S^{max}_{a_i}$, however, no agent can be seeded in the balanced graph. For $1<v_c^a<v_l^{s}$, agents in both graphs can be seeded up to $\|\vec{S}^{max}_a\|_1$.
\end{proof}

It is easy to see that Proposition \ref{compare seeding} presents very similar results as Proposition \ref{compare seeding nash}.
%\begin{rem}
%If we assume that $q_a = q_b = q$ in Proposition \eqref{compare seeding}, then we can see that the results of Proposition \eqref{compare seeding} is similar to the result of Proposition \eqref{compare seeding nash}. Using \eqref{condition symmetric} and according to Proposition \eqref{compare seeding}, we can see that for some functions $f_1, f_2, f_3, g_1, g_2, g_3$, if $q \in (f_1(\bar{v}),g_1(v_l^{s}))$, a balanced graph has a larger seeding capacity than a star graph, and if $q \in (f_2(v_h^{s} ),g_2(\bar{v}))$, a star graph has a larger seeding capacity than a balanced graph. For $q \in (f_3(v_l^{s} ),g_3(1))$ they have the same seeding capacity. Similarly, we can see that in Proposition \eqref{compare seeding nash}, $K$ plays exactly the same the role as $q$  for seeding budget for some other functions $\tilde{f}_1, \tilde{f}_2, \tilde{f}_3, \tilde{g}_1, \tilde{g}_2, \tilde{g}_3$.
%\end{rem}
%Although star and balanced graphs give us useful insights for seeding capacity in general graphs, extreme seeding capacities do not necessary happen in those graphs.
The next proposition provides us with a lower and an upper bound for minimum and maximum seeding capacities.
% which do not necessary happen in star or balanced graphs.
\begin{prop} \label{max seeding}
If $ 1 < v_c^a < v_h^{s}$, the maximum seeding capacity is given by  \begin{equation*}
\|\vec{S}^{*}_a\|^{max}_{1} = \sum_{i=1}^{k} S^{max}_{a_i},
\end{equation*}
where
\begin{equation} \label{k}
k = \min \{\lfloor \frac{n\delta}{(v_c^a-1)(2\beta - \delta)} \rfloor, n\}.
\end{equation}
On the other hand, the minimum seeding capacity is $S_{a_n}^{max}+S_{a_{n-1}}^{max}$ if $1<v_c^a<v_l^{s}$, is $S^{max}_{a_n}$ if $ v_l^{s} < v_c^a < \bar{v}$, and is zero if $\bar{v} < v_c^a < v_h^{s}$.
\end{prop}
\begin{proof}
From condition \eqref{condition} the more agents with centralities above the threshold $v_c^{a}$, the more seeding budget can be allocated. Therefore, the maximum number of $k$ agents with centralities above the threshold $v_c^{a}$ must be found. Since $v_i \geq 1$ for all agents, first a centrality of one is given to each agent and then the remainder of the centrality sum is distributed among maximum number of agents so that each agent receives at least $v_c^{a}-1$, making its overall centrality greater than $v_c^{a}$. It is easy to see that the number of such agents is upper bounded by $\lfloor \frac{\frac{2\beta n}{2\beta- \delta} - n}{v_c^a-1}\rfloor$.
This along with the fact that $1\leq k\leq n$ results in \eqref{k}.
Note that, in order to complete the proof, we should also provide an example achieving this maximum capacity. For $k=1$, the maximum seeding capacity is clearly achieved by the star graph with the seeding capacity of $S_{a_1}^{max}$. For $k\geq2$, a graph with largest seeding capacity is the one with $k$ central agents having the largest demand capacities and with equal centralities of
\begin{equation*}
\tilde{v}_h^{s} = \frac{n\delta}{k(2\beta - \delta)} + 1,
\end{equation*}
where $k$ is given in \eqref{k}, and the remainder $n-k$ agents with the minimum centrality of $\tilde{v}_l^{s}=1$.
For the graph with minimum seeding capacity, similar to the proof of Proposition \ref{compare seeding}, we have minimum seeding capacity of $S^{max}_{a_n}$ in star graph if $ v_l^{s} < v_c^a < \bar v$, and zero in balanced graph if $ \bar v< v_c^a <v_h^{s}$.
For the case where $1<v_c^a<v_l^{s}$, let $i$ be the agent with the highest centrality. Clearly, $v_l^{s}<v_i\leq v_h^{s}$. Now, considering the fact that sum of the centralities is fixed, there is an agent $j\in V-\{i\}$ for which $v_l^{s}\leq v_j$. This implies that there exist at least two agents whose centralities are above $v_c^a$. An example of a graph with exactly two centralities above $v_c^a$ is a directed star graph where all edges are directed towards the center except one edge which goes both ways.
\end{proof}

%\begin{rem}
%If we assume $q_a = q_b = q$ in Proposition \ref{max seeding}, then the graphs with the maximum and minimum seeding capacities in Proposition \ref{max seeding} are similar to the graphs with the maximum and minimum seeding budgets in Proposition \ref{max seeding nash}. In both propositions, a graph with the maximum seeding capacity or budget is a star or a $l$-star graph, and a graph with the minimum seeding capacity or budget is a star or a balanced or a graph with $n-2$ agents with centrality of one depending on $K$ in Proposition \ref{max seeding nash} and $q$ in Proposition \ref{max seeding}.
%\end{rem}

It can be seen from both Proposition \ref{max seeding} and Proposition \ref{max seeding nash} that graphs with similar structures attain maximum and minimum seeding in both scenarios.

\begin{ex} \label{ex2}
As a numerical example for the minimum and maximum seeding capacities, we consider a network with $n=15$ agents with demand capacity vectors of $\vec{S}_a^{max}=\vec{S}_b^{max}=(0.5)\vec{1}$, qualities of $q_a=q_b=1$, quality and seeding costs of $c_s=c_q=1$ and parameters of $\alpha = \beta = 1$ and $\delta=0.5$. For this example from equations \eqref{condition} and \eqref{condition2} we have $v_c^{a}=v_c^{b}= 2.5$ and as a result, from equation \eqref{k} we get $k=3$. Therefore, a graph with the maximum seeding capacity is a $3$-star with seeding capacity of $1.5$ as illustrated in Fig.~\ref{fig:3star}. Also, since $ \bar v=\frac{4}{3} < v_c^a, v_c^b < v_h^{s} = 4.8 $, a balanced graph has the minimum seeding capacity of zero. A star graph has a seeding capacity of $0.5$ which is neither a minimum nor a maximum.
\end{ex}

Note that in both Example \ref{ex1} and Example \ref{ex2} a graph with the maximum seeding budget and capacity is a $3$-star graph and a graph with the minimum seeding budget and capacity is a balanced graph. A star graph has neither a minimum nor a maximum seeding budget and capacity in both examples.
%\begin{figure}
%\centering
%\includegraphics[scale=2]{3star.pdf}
%\caption{A graph with maximum seeding capacity for Example~1}\label{fig:3star}
%\end{figure}

As we saw, the structure of the graphs with minimum and maximum seeding capacity depends on the threshold value of $v_c^{a}$. However, for certain values of $v_c^{a}$ the seeding capacity will be independent of the structure of the graph, as described in the next proposition.
\begin{prop} \label{indep of network}
If $v_c^a > v_h^{s}$ no graph can be seeded. On the other hand, if $v_c^a < 1$ all graphs can be seeded equally up to agents' maximum demand capacities.
\end{prop}
\begin{proof}
The maximum possible centrality happens for the central agent of the star graph as shown in Lemma~\ref{v_Lv_H}. As a result, if $v_c^a > v_h^{s}$, then we have $v_c^a >  v_i$ for all $i$ in any graph and from condition \eqref{condition} no agent can be seeded. Also, since from definition $1 \leq v_i$ for all $i$, if $ v_c^{a} < 1$ then we have $ v_c^{a} < v_i$ for all agents in any graph, and any graph can be seeded up to agents' maximum demand capacities, given the availability of budget.
\end{proof}

It is easy to see that Proposition \ref{indep of network} presents very similar results as Proposition \ref{indep of network nash}.

\section{Conclusion} \label{sec5}

We proposed and studied a strategic model of marketing and product consumption in social networks. Two firms %providing products with preset qualities
compete for maximizing the consumption of their products in a social network.
Initially firms utilize a limited budget to either design the quality of their products or initially seed a set of agents in the social network.
Agents are myopic yet utility maximizing, given the qualities of the products and actions of their neighbors. This myopic best response results in a local, linear update dynamics for the consumptions of the agents.
%After the decision of firms, agents choose their consumptions following a myopic best response dynamics which results in a local, linear update for the consumptions.
We characterized the unique Nash equilibrium of the game between firms. We showed that at the Nash equilibrium, firms invest more budget on quality when their budgets are close. However, as the difference between budgets increases, firms spend more budget on seeding. We also showed that given equal budget of firms, if seeding budget is nonzero for a balanced graph it will also be nonzero for any other graph, and if seeding budget is zero for a star graph it will be zero for any other graph too.
Afterwards, we considered a different scenario in which firms produce two products with some preset qualities that can only be improved marginally.
%At some point in time, firms learn about the network structure receive a limited budget which they can use to trigger a larger consumption of their products in the network.
At some point in time, firms spend a limited budget to marginally improve the quality of their products and to give free offers to a set of agents in the network in order to promote their products. We derived a simple threshold rule for the optimal allocation of the budget between new seedings and quality improvement. We showed that the optimal allocation of the budget in particular depends on the entire centrality distribution of the graph and the qualities of the products.
%In particular, we showed that a graph with a higher number of agents with centralities above a certain threshold, has a higher seeding budget in the optimal allocation.
Furthermore, we derived similar results to the original setup for this scenario, in which preset qualities resemble
the role of budgets.

\bibliographystyle{IEEEtran}
\bibliography{CDC2014_REF}
\end{document}